\documentclass[aps,pra,superscriptaddress,showpacs]{revtex4}

\usepackage{amsfonts,amssymb,amsmath}
\usepackage{amsthm}
\usepackage{graphics,graphicx,epsfig}
\usepackage[dvipsnames]{xcolor}
\usepackage[normalem]{ulem}
\usepackage{comment}
\usepackage{epstopdf}
\usepackage{subfigure}
\usepackage{float}
\usepackage{subfigure}
\usepackage[colorlinks=true, pdfborder=001, linkcolor=blue, anchorcolor=blue, citecolor=blue, urlcolor=blue]{hyperref}

\newtheorem{theorem}{Theorem}

\newtheorem{remark}{Remark}
\newtheorem{example}{Example}
\newtheorem{lemma}{Lemma}
\newcommand{\ket}[1]{|#1\rangle}
\newcommand{\bra}[1]{\langle#1|}

\def\be{\begin{eqnarray}}
\def\ee{\end{eqnarray}}
\def\Tr{\mathrm{Tr}}

\begin{document}
\title{Tighter monogamy and polygamy relations for a superposition of the generalized $W$-class state and vacuum}

\author{Le-Min Lai}
\affiliation{School of Mathematical Sciences,  Capital Normal University,  Beijing 100048,  China}
\author{Shao-Ming Fei}
\email{feishm@cnu.edu.cn}
\affiliation{School of Mathematical Sciences,  Capital Normal University, Beijing 100048,  China}
\affiliation{Max Planck Institute for Mathematics in the Sciences, 04103 Leipzig, Germany}
\author{Zhi-Xi Wang}
\email{wangzhx@cnu.edu.cn}
\affiliation{School of Mathematical Sciences, Capital Normal University, Beijing 100048,   China}


\begin{abstract}
Monogamy and polygamy relations characterize the distributions of entanglement in multipartite systems. We investigate the monogamy and polygamy relations with respect to any partitions for a superposition of the generalized $W$-class state and vacuum in terms of the Tsallis-$q$ entanglement and the R\'enyi-$\alpha$ entanglement. By using the Hamming weight of the binary vectors related to the partitions of the subsystems, new classes of monogamy and polygamy inequalities are derived, which are shown to be tighter than the existing ones. Detailed examples are presented to illustrate the finer characterization of entanglement distributions.

\vspace{0.5cm}
\end{abstract}

\maketitle

\section{Introduction}\label{sec1}

Quantum entanglement \cite{Ent1,Ent2,Ent3} is a quintessential feature of quantum mechanics which distinguishes the quantum from the classical world and plays an important role in quantum information processing. One distinguished property of quantum entanglement without any classical
counterpart is its limited shareability in multipartite quantum systems, known as the monogamy of entanglement (MoE) \cite{BMT,JSK}.
MoE is the fundamental ingredient in many quantum information processing tasks such as the security proof in quantum cryptographic scheme \cite{BCH}
and the security analysis of quantum key distribution \cite{MP}.

For a tripartite quantum state $\rho_{ABC}$ with its reduced
density matrices $\rho_{AB}=\Tr_C {\rho_{ABC}}$ and $\rho_{AC}=\Tr_B \rho_{ABC}$, mathematically MoE can be characterized in terms of some bipartite entanglement measure $\varepsilon$ as
$\varepsilon(\rho_{A|BC})\geqslant\varepsilon(\rho_{AB})+\varepsilon(\rho_{AC})$, where $\varepsilon(\rho_{A|BC})$ denotes the shared entanglement between subsystems $A$ and $BC$, which measures the degree of entanglement between $A$ and $BC$,  and $\varepsilon(\rho_{AB})$ ($\varepsilon(\rho_{AC})$) is the bipartite
entanglement between A and B (A and C). This inequality  conveys the MoE principle that the amount of
entanglement shared between $A $ and $B$ restricts the
possible amount of entanglement between $A$ and $C$ so
that their sum does not exceed the total bipartite
entanglement between $A$ and the composite $BC$
system. Note that the monogamy inequalities provide an upper bound for bipartite sharability of entanglement in a multipartite system. It is also known that the
assisted entanglement $\varepsilon^a$ \cite{GG12007,GG2007}, which is a dual amount to bipartite
entanglement measures, has a dually monogamous property
in multipartite systems. This dually monogamous
property of entanglement is also characterized as a polygamy inequality, which is quantitatively displayed as $\varepsilon^a(\rho_{A|BC})\leqslant\varepsilon^a(\rho_{AB})+\varepsilon^a(\rho_{AC})$ for a tripartite system, where $\varepsilon^a(\cdot)$ is the corresponding entanglement measure of assistance associated to $\varepsilon$. Similarly, a polygamy inequality sets a lower bound for the distribution  of bipartite entanglement in multipartite systems.

The first monogamy relation was proven by Coffman {\it et al.} \cite{CV}  based on the squared concurrence for arbitrary three-qubit states, known as the CKW inequality. Later, it was generalized to multipartite systems \cite{OTJ,BYK}. Besides concurrence, the monogamy  relations are also given by various entanglement measures for multipartite systems \cite{HT,AG,Kim2009,BYK2,ZXN,JZX,Kim2016T1,JZX1,LY,JZX2,JZX5,Kim2018}. The polygamy relation was first obtained in terms of the tangle of assistance for three-qubit systems \cite{GG12007}, and then generalized to multiqubit systems and arbitrary dimensional multipartite systems \cite{Kim2018,Kim2010,BF2009,Kim2012,JZX3,KJS20181,KJS20182,Shi1}.

However, it was found that the CKW inequality is invalid for higher-dimensional systems \cite{OYC}. In \cite{LC} the authors discovered that in some higher-dimensional systems there is no nontrivial monogamy relation satisfied by any additive entanglement measures. It seems that only the squashed entanglement satisfies the monogamy relation for arbitrary dimensional systems \cite{CM}. Therefore, the MoE for high dimensional systems has attracted much attention.

In \cite{Kim2008} Kim {\it et al.} proved that the $n$-qudit generalized $W$-class (GW) states satisfy the monogamy inequality in terms of the squared concurrence. In \cite{Choi2015} Choi and Kim showed that the superposition of the generalized $W$-class states and vacuum (GWV) states satisfy the strong monogamy inequality based on the squared convex roof extended negativity. In \cite{Kim2016} Kim focused on a large class of mixed states that are in a partially coherent superposition of a generalized $W$-class state and the vacuum, and showed that those states obey the strong monogamy inequality by using the squared convex roof extended negativity. Very recently,  Shi {\it et al.} presented in \cite{Shi2020} new monogamy and polygamy relations with respect to any partition for $n$-qudit GWV states by using the analytical formula of the Tsallis-$q$ entanglement (T$q$E) \cite{Kim2010}. Moreover,  Liang {\it et al.} presented in \cite{Liang2020} the monogamy and polygamy relations  for GWV states in terms of the R\'enyi-$\alpha$ entanglement (R$\alpha$E) \cite{Renyil,KimR2010}. Inspired by these developments, we  investigate further the monogamy and polygamy relations for the GWV states in high dimensional quantum systems.

In this paper, by using the Hamming weight of the binary vector associated with the distribution of subsystems, we establish a class of  monogamy  and polygamy relations for the GWV states based on  T$q$E and R$\alpha$E. We derive monogamy and polygamy inequalities which are tighter than  those given in \cite{Shi2020,Liang2020}, thus giving rise to finer characterizations of the entanglement distributions among the high dimensional quantum subsystems for the GWV states.

\section{Tighter  monogamy and polygamy relations  based on T$q$E and T$q$EoA for GWV states}\label{sec3}

A class of $n$-qubit $W$-class states and $n$-qudit generalized $W$-class states are, respectively, defined by
\begin{eqnarray}
	\ket{\psi}_{A_1 A_2 ... A_n}=&
	a_1 \ket{10\cdots0}+a_2 \ket{01\cdots0}+...+a_n \ket{00\cdots1}
	\label{qubitWstate}
\end{eqnarray}
and
\begin{eqnarray}
	\left|W_n^d \right\rangle_{A_1\cdots A_n}=\sum_{i=1}^{d-1}(&a_{1i}{\ket {i0\cdots 0}} +a_{2i}{\ket {0i\cdots 0}}+\cdots +a_{ni}{\ket {00\cdots 0i}}),
	\label{quditWstate}
\end{eqnarray}
where  $\sum_{i=1}^{n}|a_i|^2 =1$ and  $\sum_{s=1}^{n}\sum_{i=1}^{d-1}|a_{si}|^2=1$. When $d=2$, (\ref{quditWstate}) reduces to the $n$-qubit $W$-class states.

Choi and Kim introduced in \cite{Choi2015} the GWV state $\ket{\psi}_{A_1\cdots A_n}$,
\begin{equation}
	\ket{\psi}_{A_1A_2\cdots A_n}=\sqrt{p}\left|W_n^d \right\rangle_{A_1\cdots A_n}+\sqrt{1-p}\ket{0\cdots 0}_{A_1\cdots A_n}
	\label{GWV}
\end{equation}
for $0\leqslant p \leqslant 1$.
Let $\rho_{A_{i_1}\cdots A_{i_{m}}}$ denotes the reduced density matrix with respect to $\ket{\psi}_{A_1\cdots A_n}$ in $m$-qudit subsystems $A_{i_1}\cdots A_{i_{m}}$ with $2 \leqslant m \leqslant  n-1$. It has been shown that for any pure state decomposition of $\rho_{A_{i_1}\cdots A_{i_{m}}}$,
\begin{equation}
	\rho_{A_{i_1}\cdots A_{i_{m}}}=\sum_{k}q_k\ket{\phi_k}_{A_{i_1}\cdots A_{i_{m}}}\bra{\phi_k},
	\label{rhoa1aj1ajm-1}
\end{equation}
$\ket{\phi_k}_{A_{i_1}\cdots A_{i_{m}}}$ is a superposition of an $m$-qudit generalized $W$-class state and vacuum \cite{Choi2015}. Moreover,
for an $n$-qudit GWV state $\ket{\psi}_{A_1A_2\cdots A_n}$ and an
arbitrary partition $P=\{P_1,\cdots,P_r\}$ of the set $S=\{A_1,\cdots,A_n\}$, $r\leqslant n$,
$P_i\cup P_j=\emptyset$ $(i\neq j)$ and $\bigcup_iP_i=S$,
the state $\ket{\psi}_{P_1P_2\cdots P_r}$ is also a GWV state \cite{Kim2016}.

The T$q$E of a bipartite pure state $|\psi\rangle_{AB}$ is defined as \cite{Kim2010}
\begin{equation}
	T_q(|\psi\rangle_{AB})=S_q(\rho_A)=\frac{1}{q-1}(1-{\rm tr}\rho_A^q),
\end{equation}
where $q>0$ and $q\neq1$.
When $q$ tends to 1, $T_q(\rho)$ converges to the von Neumann entropy, i.e., $\lim\limits_{q\rightarrow 1}T_q(\rho)=-{\rm tr}\rho\log_2\rho=S(\rho)$.
The Tsallis-$q$ entanglement of a bipartite mixed state $\rho_{AB}$ is given by
\begin{equation}
	T_q(\rho_{AB})=\min\limits_{\{p_i,|\psi_i\rangle\}}\sum\limits_{i}p_iT_q(|\psi_i\rangle)
\end{equation}
with the minimum taken over all possible pure state decompositions of $\rho_{AB}$. As a dual concept of T$q$E, its Tsallis-$q$ entanglement of assistance (T$q$EoA) is defined as
\begin{equation}
	T_q^a(\rho_{AB})=\max\limits_{\{p_i,|\psi_i\rangle\}}\sum\limits_{i}p_iT_q(|\psi_i\rangle)
\end{equation}
with the maximum taken over all possible pure state decompositions of $\rho_{AB}$.
There is an analytic relationship
between the Tsallis-$q$ entanglement and concurrence \cite{Concurrence1,Concurrence2} for $q\in[\frac{5-\sqrt{13}}{2}, \frac{5+\sqrt{13}}{2}]$ \cite{Yuan2016},
\begin{equation}\label{T1}
	T_q(|\psi\rangle_{AB})=g_q(\mathcal{C}^2(|\psi\rangle_{AB})),
\end{equation}
where
\begin{equation}\label{T2}
	g_q(x)=\frac{1}{q-1}\Big[1-\Big(\frac{1+\sqrt{1-x}}{2}\Big)^q-\Big(\frac{1-\sqrt{1-x}}{2}\Big)^q\Big].
\end{equation}
It has also been shown that $T_q(|\psi\rangle)=g_q(\mathcal{C}^2(|\psi\rangle))$ for any  $2\otimes m~(m\geqslant2)$ pure state $|\psi\rangle$,
and $T_q(\rho)=g_q(\mathcal{C}^2(\rho))$ for 2-qubit mixed state $\rho$ \cite{Kim2010}.
Therefore,  (\ref{T1}) holds for any $q$ such that $g_q(x)$ in (\ref{T2}) is monotonically increasing and convex.

Let $\rho_{A_{i_1}A_{i_2}\cdots A_{i_m}}$ be the reduced density matrix of $\ket{\psi}_{A_1\cdots A_n}$ in (\ref{GWV}). Denote by $\{P_1,\cdots,P_r\}$  a partition of the set $\{A_{i_1},A_{i_2},\cdots,A_{i_m}\}$, $r\leqslant m\leqslant n$. In \cite{Shi2020} Shi {\it et al.} proved that
\begin{eqnarray}\label{T8}
	T_q(\rho_{A_{i_1}|A_{i_2}\cdots A_{i_m}})=g_q(C^2(\rho_{A_{i_1}|A_{i_2}\cdots A_{i_m}}))
\end{eqnarray}
for $q\in[\frac{5-\sqrt{13}}{2},\frac{5+\sqrt{13}}{2}]$, and
\begin{eqnarray}\label{T7}
	T_q^a(\rho_{A_{i_1}|A_{i_2}\cdots A_{i_m}})=T_q(\rho_{A_{i_1}|A_{i_2}\cdots A_{i_m}})=g_q(C^2(\rho_{A_{i_1}|A_{i_2}\cdots A_{i_m}}))
\end{eqnarray}
for $q\in[\frac{5-\sqrt{13}}{2}, 2]\cup[3, \frac{5+\sqrt{13}}{2}]$.

Furthermore,  for GWV states the authors in \cite{Shi2020} established the monogamy relation
\begin{eqnarray}\label{T3}
	T_q^\mu(\rho_{P_1|P_2\cdots P_r})\geqslant \sum_{j=2}^{r} T_q^\mu(\rho_{P_1P_j})
\end{eqnarray}
for $q\in[\frac{5-\sqrt{13}}{2},\frac{5+\sqrt{13}}{2}]$ and $\mu\in[2,\infty)$, and the general monogamy relation,
\begin{eqnarray}\label{T4}
	T_q^\gamma(\rho_{P_1|P_2\cdots P_r})&\geqslant& \sum_{j=2}^{t}(2^{\frac{\gamma}{\mu}}-1)^{j-2}T_q^\gamma(\rho_{P_1P_j})+(2^{\frac{\gamma}{\mu}}-1)^t\sum_{j=t+1}^{r-1}T_q^\gamma(\rho_{P_1P_j})\nonumber\\
	&&\ \ +(2^{\frac{\gamma}{\mu}}-1)^{t-1}T_q^\gamma(\rho_{P_1P_r}),
\end{eqnarray}
conditioned that $T_q(\rho_{P_1P_i})\leqslant T_q(\rho_{P_1|P_{i+1}\cdots P_r}) $ for $i=2,3,\cdots, t$, and $T_q(\rho_{P_1P_j})\geqslant T_q(\rho_{P_1|P_{j+1}\cdots P_r})$ for $j=t+1,\cdots, r-1$ with $\gamma\in[0,\mu]$ and $\mu\in[2,\infty)$.

For $q\in[\frac{5-\sqrt{13}}{2}, 2]\cup[3, \frac{5+\sqrt{13}}{2}]$, the following polygamy relation based on the T$q$EoA has been obtained \cite{Shi2020},
\begin{eqnarray}\label{T5}
	(T_q^a(\rho_{P_1|P_2\cdots P_k}))^\mu\leqslant \sum_{j=2}^{r} (T_q^a(\rho_{P_1P_j}))^\mu
\end{eqnarray}
with $\mu\in(0,1]$.

Next, we provide a class of monogamy and  polygamy  inequalities which are tighter than  inequalities  (\ref{T3}), (\ref{T4}) and (\ref{T5}), respectively, by using the following lemma \cite{Yang2019}.

\begin{lemma}\label{lem1}
	For any real numbers $x, k$ and $t$, we have
	
	$\mathrm{(a)}$ $(1+x)^t\geqslant 1+\frac{(1+k)^t -1}{k^t} x^t$ for $0\leqslant x \leqslant k\leqslant1$, $t\geqslant1$;
	
	$\mathrm{(b)}$ $(1+x)^t\geqslant 1+\frac{(1+k)^t -1}{k^t} x^t$ for $x\geqslant k\geqslant 1$,  $0\leqslant t\leqslant1$;
	
	$\mathrm{(c)}$ $(1+x)^t \leqslant 1+\frac{(1+k)^t -1}{k^t}x^t$ for $0\leqslant x \leqslant k\leqslant1$, $0\leqslant t\leqslant1$.
	
\end{lemma}

\subsection{Tighter  monogamy relations in terms of T$q$E}\label{section3-1}

For any nonnegative integer $j$ and its binary expansion $j=\sum\limits_{i=0}^{s-1} j_i 2^i$, with $\log_{2}j \leqslant s$ and $j_i \in \{0, 1\}$ for $i=0, \cdots, s-1$,
one can define a unique binary vector
$\overrightarrow{j}=\left(j_0,~ j_1,~\cdots,~j_{s-1}\right)$. The Hamming weight $\omega_{H}(\overrightarrow{j})$ of $\overrightarrow{j}$ is defined as the number of $1's$ in $\{j_0,~j_1,~\cdots,~j_{s-1}\}$.
The Hamming weight $\omega_{H}(\overrightarrow{j})$ is bounded above by $\log_{2}j$,
\begin{equation}\label{weight}
	\omega_{H}\left(\overrightarrow{j}\right)\leqslant \log_{2}j \leqslant j.
\end{equation}

In the following we denote by $\rho_{A_{i_1}A_{i_2}\cdots A_{i_{m}}}$ the reduced density matrices of a GWV state $\ket{\psi}_{A_1\cdots A_n}$ given in (\ref{GWV}), and $\{P,P_0,P_1,\cdots,P_{r-1}\}$ a partition of the set $\{A_{i_1},A_{i_2},\cdots,A_{i_m}\}$, $r\leqslant m-1\leqslant n-1$.

\begin{theorem}\label{thm10}
	If \begin{equation}\label{thm10:1}
		k T_q^2(\rho_{PP_j})\geqslant T_q^2(\rho_{PP_{j+1}})\geqslant 0
	\end{equation}
	for $j=0,1,\cdots, r-2$ and $0<k\leqslant 1$, we have
	\begin{eqnarray}\label{thm10:2}
		T_q^\beta(\rho_{P|P_0\cdots P_{r-1}})\geqslant\sum\limits_{j=0}^{r-1}(\mathcal{K}_\beta)^{\omega_H(\overrightarrow{j})}T_q^\beta(\rho_{ PP_j}),
	\end{eqnarray}
	where $\beta\in[2, \infty)$, $q\in [\frac{5-\sqrt{13}}{2},\frac{5+\sqrt{13}}{2}]$, $\mathcal{K}_\beta=\frac{(1+k)^\frac{\beta}{2}-1}{k^\frac{\beta}{2}}$.
\end{theorem}

\begin{proof}
	From inequality (\ref{T3}), one has $T_q^2(\rho_{P|P_0\cdots P_{r-1}})\geqslant \sum\limits_{j=0}^{r-1} T_q^2(\rho_{PP_j})$. Thus, it is sufficient to show that
	\begin{equation}\label{thm10:3}
		\Bigg(\sum\limits_{j=0}^{r-1} T_q^2(\rho_{PP_j})\Bigg)^{\frac{\beta}{2}}
		\geqslant\sum\limits_{j=0}^{r-1}(\mathcal{K}_\beta)^{\omega_H(\overrightarrow{j})}T_q^\beta(\rho_{ PP_j}).
	\end{equation}
	
	First, we prove that the inequality (\ref{thm10:3}) holds for the case of $r=2^s$ by using mathematical induction on $s$.
	For $s=1$, using Lemma \ref{lem1} $\mathrm{(a)}$, we have
	\begin{eqnarray}\label{thm10:4}
		(T_q^2(\rho_{PP_0})+T_q^2(\rho_{PP_1}))^{\frac{\beta}{2}}
		&=&T_q^\beta(\rho_{PP_0}) \Big(1+\frac{T_q^2(\rho_{PP_1})}{T_q^2(\rho_{PP_0})}\Big)^{\frac{\beta}{2}} \nonumber \\
		&\geqslant& T_q^\beta(\rho_{PP_0}) \Bigg[1+\mathcal{K}_\beta\Bigg(\frac{T_q^\beta(\rho_{PP_1})}{T_q^\beta(\rho_{PP_0})}\Bigg)\Bigg] \nonumber \\
		&=&T_q^\beta(\rho_{PP_0})+\mathcal{K}_\beta T_q^\beta(\rho_{PP_1}).
	\end{eqnarray}
	Thus, the inequality (\ref{thm10:3}) holds for $s=1$.
	
	Assume that the inequality (\ref{thm10:3}) holds for $r=2^{s-1}$ with $s\geqslant 2$. Consider the case of $r=2^s$.
	From  (\ref{thm10:1}) we have $T_q^2(\rho_{PP_{j+2^{s-1}}})\leqslant k^{2^{s-1}}T_q^2(\rho_{PP_j})$ for $j=0,1,\cdots, 2^{s-1}-1$.
	Therefore,
	\begin{equation*}
		\frac{\sum\nolimits_{j=2^{s-1}}^{2^s-1}T_q^2(\rho_{PP_j})}{\sum\nolimits_{j=0}^{2^{s-1}-1}
			T_q^2(\rho_{PP_j})}\leqslant k^{2^{s-1}}\leqslant k\leqslant1.
	\end{equation*}
	
	Again using Lemma \ref{lem1} $\mathrm{(a)}$, we have
	\begin{eqnarray}\label{thm10:5}
		\Bigg(\sum\limits_{j=0}^{2^s-1}T_q^2(\rho_{PP_j})\Bigg)^{\frac{\beta}{2}}
		&=&\Bigg(\sum\limits_{j=0}^{2^{s-1}-1}T_q^2(\rho_{PP_j})\Bigg)^{\frac{\beta}{2}}
		\Bigg(1+\frac{\sum_{j=2^{s-1}}^{2^s-1}T_q^2(\rho_{PP_j})}{\sum_{j=0}^{2^{s-1}-1}T_q^2
			(\rho_{PP_j})}\Bigg)^{\frac{\beta}{2}} \nonumber\\
		&\geqslant& \Bigg(\sum\limits_{j=0}^{2^{s-1}-1}T_q^2(\rho_{PP_j})\Bigg)^{\frac{\beta}{2}}
		\Bigg[1+\mathcal{K}_\beta\Bigg(\frac{\sum_{j=2^{s-1}}^{2^s-1}T_q^2(\rho_{PP_j})}{\sum_{j=0}^{2^{s-1}-1}T_q^2
			(\rho_{PP_j})}\Bigg)^{\frac{\beta}{2}}\Bigg] \nonumber\\
		&=&\Bigg(\sum\limits_{j=0}^{2^{s-1}-1}T_q^2(\rho_{PP_j})\Bigg)^{\frac{\beta}{2}}+\mathcal{K}_\beta\Bigg(\sum\limits_{j=2^{s-1}}^{2^s-1}T_q^2(\rho_{PP_j})\Bigg)
		^{\frac{\beta}{2}}.
	\end{eqnarray}
	From the induction hypothesis, we have
	\begin{equation}\label{thm10:6}
		\Bigg(\sum\limits_{j=0}^{2^{s-1}-1}T_q^2(\rho_{PP_j})\Bigg)^{\frac{\beta}{2}} \geqslant
		\sum\limits_{j=0}^{2^{s-1}-1}(\mathcal{K}_\beta)^{\omega_H(\overrightarrow{j})}
		T_q^\beta(\rho_{PP_j}).
	\end{equation}
	By relabeling the subsystems, we can easily get
	\begin{equation}\label{thm10:7}
		\Bigg(\sum\limits_{j=2^{s-1}}^{2^s-1}T_q^2(\rho_{PP_j})\Bigg)^{\frac{\beta}{2}} \geqslant
		\sum\limits_{j=2^{s-1}}^{2^s-1}(\mathcal{K}_\beta)^{\omega_H(\overrightarrow{j})-1}T_q^\beta(\rho_{PP_j}).
	\end{equation}
	From inequality (\ref{thm10:5}) together with inequalities (\ref{thm10:6}) and (\ref{thm10:7}), we have
	\begin{equation}
		\Bigg(\sum\limits_{j=0}^{2^s-1}T_q^2(\rho_{PP_j})\Bigg)^{\frac{\beta}{2}}\geqslant
		\sum\limits_{j=0}^{2^s-1}(\mathcal{K}_\beta)^{\omega_H(\overrightarrow{j})}T_q^\beta(\rho_{PP_j}).
	\end{equation}
	
	Now we extend the above conclusion to arbitrary integer $r$. Note that there always exists some $s$ such that $0<r\leqslant 2^s$. Let us consider a $(2^s+1)$-partite quantum state,
	\begin{equation}\label{thm10:8}
		\gamma_{PP_0P_1\ldots P_{2^s-1}}=\rho_{PP_0P_1\cdots P_{r-1}}\otimes \sigma_{P_r\cdots P_{2^s-1}},
	\end{equation}
	which is the tensor product of $\rho_{PP_0P_1\cdots P_{r-1}}$ and an arbitrary $(2^s-r)$-partite state $\sigma_{P_r\cdots P_{2^s-1}}$.
	As just proved above for this state, we have
	\begin{equation}
		T_q^\beta(\gamma_{P|P_0P_1\cdots P_{2^s-1}})
		\geqslant\sum\limits_{j=0}^{2^s-1}(\mathcal{K}_\beta)^{\omega_H(\overrightarrow{j})}T_q^\beta(\gamma_{PP_j}),
	\end{equation}
	where $\gamma_{PP_j}$ is the reduced density matrix of $\gamma_{PP_0P_1\cdots P_{2^s-1}}$, $j=0,1,\cdots,2^s-1$.
	
	Taking into account the following obvious facts: $T_q\left(\gamma_{P|P_0 P_1 \cdots P_{2^s-1}}\right)=T_q\left(\rho_{P|P_0 P_1 \cdots P_{r-1}}\right)$,
	$T_q\left(\gamma_{PP_j}\right)=0$ for $j=r, \cdots , 2^s-1$,
	and $T_q(\gamma_{PP_j})=T_q(\rho_{PP_j})$ for each $j=0, \cdots , r-1$, we get
	\begin{eqnarray}
		T_q^\beta(\rho_{P|P_0P_1\cdots P_{r-1}}) &=&T_q^\beta(\gamma_{P|P_0P_1\cdots P_{2^s-1}}) \nonumber\\
		&\geqslant& \sum\limits_{j=0}^{2^s-1}(\mathcal{K}_\beta)^{\omega_H(\overrightarrow{j})}T_q^\beta(\gamma_{PP_j}) \nonumber\\
		&= &\sum\limits_{j=0}^{r-1}(\mathcal{K}_\beta)^{\omega_H(\overrightarrow{j})}T_q^\beta(\rho_{PP_j}).
	\end{eqnarray}
	This completes the proof.
\end{proof}

\begin{remark} Since $(\mathcal{K}_\beta)^{\omega_H(\overrightarrow{j})}\geqslant 1$ for  any $\beta\geqslant2$, we have
	\begin{equation}\label{re3}
		T_q^\beta(\rho_{P|P_0\cdots P_{r-1}})\geqslant\sum\limits_{j=0}^{r-1}(\mathcal{K}_\beta)^{\omega_H(\overrightarrow{j})}T_q^\beta(\rho_{ PP_j})\geqslant\sum\limits_{j=0}^{r-1}T_q^\beta(\rho_{ PP_j}).
	\end{equation}
	Therefore, we provide a monogamy relation based on TqE with larger lower bound than  (\ref{T3}) in Ref.~\cite{Shi2020} .
	
	In the Theorem \ref{thm10}	when $\kappa=1$,  for  any GWV states in the order of  the partitions $P_0, P_1, \cdots, P_{r-1}$ satisfying   $ T_q(\rho_{PP_j})\geqslant T_q(\rho_{PP_{j+1}})\geqslant 0$, $j=0,1,\cdots, r-2$ ,  we get
	\begin{equation}
		T_q^\beta(\rho_{P|P_0\cdots P_{r-1}})\geqslant\sum\limits_{j=0}^{r-1}(2^{\frac{\beta}{2}}-1)^{\omega_H(\overrightarrow{j})}T_q^\beta(\rho_{ PP_j}),
	\end{equation}
	which is tighter than the inequality
	(\ref{T3}) in \cite{Shi2020}.
	When  $0<\kappa<1$, for the GWV states  satisfying certain conditions (\ref{thm10:1}),  we can also improve the monogamy  relations  of Ref. \cite{Shi2020}  from (\ref{re3}).
	
	Furthermore, as  $\frac{(1+k)^\frac{\beta}{2}-1}{k^\frac{\beta}{2}}$ is a decreasing  function of $k$ for $k\in(0,1]$, $\beta\geqslant2$, the inequality (\ref{thm10:2}) gets tighter as $k$ decreases.
\end{remark}

\begin{theorem}\label{thm11}
	When $q\in [\frac{5-\sqrt{13}}{2},\frac{5+\sqrt{13}}{2}]$, we have
	\begin{eqnarray}\label{thm11:1}
		T_q^\beta(\rho_{P|P_0\cdots P_{r-1}})\geqslant\sum\limits_{j=0}^{r-1}(\mathcal{K}_\beta)^{j}T_q^\beta(\rho_{ PP_j})
	\end{eqnarray}
	conditioned that
	\begin{equation}\label{thm11:2}
		k T_q^2(\rho_{PP_l})\geqslant \sum\limits_{j=l+1}^{r-1}T_q^2(\rho_{PP_{j}})
	\end{equation}
	for $l=0,1,\cdots, r-2$, $0<k\leqslant1$, where  $\beta\in[2,\infty)$ and $\mathcal{K}_\beta=\frac{(1+k)^\frac{\beta}{2}-1}{k^\frac{\beta}{2}}$.
\end{theorem}

\begin{proof}
	From inequality (\ref{T3}),  we only need to prove
	\begin{equation}\label{thm11:3}
		\Big(\sum_{j=0}^{r-1}T_q^2\left(\rho_{PP_j}\right)\Big)^{\frac{\beta}{2}}
		\geqslant \sum_{j=0}^{r-1} ( \mathcal{K}_\beta)^{j}T_q^\beta(\rho_{PP_j}).
	\end{equation}
	We use mathematical induction on $r$ here. It is obvious that inequality (\ref{thm11:3}) holds for $r=2$ from (\ref{thm10:4}). Assume that it also holds for any positive integer less than $r$. Since $\frac{\sum\limits_{j=1}^{r-1}T_q^2(\rho_{PP_{j}})}{T_q^2(\rho_{PP_0})}\leqslant k$, we have
	\begin{eqnarray}
		\left(\sum_{j=0}^{r-1}T_q^2\left(\rho_{PP_j}\right)\right)^{\frac{\beta}{2}}
		&=&T_q^\beta(\rho_{PP_0})
		\Bigg(1+\frac{\sum_{j=1}^{r-1}T_q^2(\rho_{PP_j})}
		{T_q^2(\rho_{PP_0})} \Bigg)^{\frac{\beta}{2}}\nonumber\\
		&\geqslant& T_q^\beta(\rho_{PP_0})\Bigg[1+ \mathcal{K}_\beta\Bigg(\frac{\sum\nolimits_{j=1}^{r-1}T_q^2(\rho_{PP_{j}})}{T_q^2(\rho_{PP_0})}\Bigg)^{\frac{\beta}{2}}\Bigg]\nonumber\\
		&=&T_q^\beta(\rho_{PP_0})+\mathcal{K}_\beta\Bigg(\sum\limits_{j=1}^{r-1}T_q^2(\rho_{PP_{j}})\Bigg)^{\frac{\beta}{2}}\nonumber\\
		&\geqslant& T_q^\beta(\rho_{PP_0})+\mathcal{K}_\beta \sum\limits_{j=1}^{r-1}(\mathcal{K}_\beta)^{j-1}T_q^\beta(\rho_{ PP_j})\nonumber\\
		&=&\sum\limits_{j=0}^{r-1}(\mathcal{K}_\beta)^{j}T_q^\beta(\rho_{ PP_j}),
	\end{eqnarray}
	where the first inequality is due to Lemma \ref{lem1} $\mathrm{(a)}$ and the second inequality is due to the induction hypothesis.
\end{proof}
According to inequality (\ref{weight}), we obtain
\begin{equation*}
	T_q^\beta(\rho_{P|P_0\cdots P_{r-1}})\geqslant\sum\limits_{j=0}^{r-1}(\mathcal{K}_\beta)^{j}T_q^\beta(\rho_{ PP_j})\geqslant\sum\limits_{j=0}^{r-1}(\mathcal{K}_\beta)^{\omega_H(\overrightarrow{j})}T_q^\beta(\rho_{ PP_j})
\end{equation*}
for $\beta\geqslant2$. Therefore the inequality (\ref{thm11:1}) of Theorem \ref{thm11} is tighter than  the inequality (\ref{thm10:2}) of Theorem \ref{thm10} under certain conditions.

In general, the conditions (\ref{thm11:2}) is not always satisfied. We derive the following monogamy inequality with different conditions.

\begin{theorem}\label{thm14}
	When $q\in [\frac{5-\sqrt{13}}{2},\frac{5+\sqrt{13}}{2}]$, we have
	\begin{eqnarray}\label{thm14:1}
		T_q^\beta(\rho_{P|P_0\cdots P_{r-1}})&\geqslant& \sum_{j=0}^{t}(\mathcal{K}_\beta)^{j}T_q^\beta(\rho_{PP_j})+(\mathcal{K}_\beta )^{t+2}\sum_{j=t+1}^{r-2}T_q^\beta(\rho_{PP_j})\nonumber\\
		&&\ \ \
		+(\mathcal{K}_\beta)^{t+1}T_q^\beta(\rho_{PP_{r-1}})
	\end{eqnarray}
	conditioned that
	$k T_q^2(\rho_{PP_i})\geqslant T_q^2(\rho_{P|P_{i+1}\cdots P_{r-1}})$ for $i=0,1,\cdots, t$ and $T_q^2(\rho_{PP_j})\leqslant k T_q^2(\rho_{P|P_{j+1}\cdots P_{r-1}})$ for $j=t+1,\cdots, r-2$, $\forall  0<k\leqslant1$, $0\leqslant t\leqslant r-3$, $r\geqslant3$, where $\beta\in[2,\infty)$ and $\mathcal{K}_\beta=\frac{(1+k)^\frac{\beta}{2}-1}{k^\frac{\beta}{2}}$.
\end{theorem}

\begin{proof}
	From Theorem \ref{thm10} for the case $r=2$, we have
	\begin{eqnarray}\label{thm14:2}
		T_q^\beta(\rho_{P|P_0\cdots P_{r-1}})
		&\geqslant& T_q^\beta(\rho_{PP_0})+\mathcal{K}_\beta T_q^\beta(\rho_{P|P_1\cdots P_{r-1}})\nonumber\\
		&\geqslant& \cdots\nonumber\\
		&\geqslant& \sum_{j=0}^{t}(\mathcal{K}_\beta)^{j}T_q^\beta(\rho_{PP_j})+(\mathcal{K}_\beta)^{t+1}T_q^\beta(\rho_{P|P_{t+1}\cdots P_{r-1}}).
	\end{eqnarray}
	Since $T_q^2(\rho_{PP_j}) \leqslant k T_q^2(\rho_{P|P_{j+1}\cdots P_{r-1}})$ for $j=t+1,\cdots, r-2$, using Theorem \ref{thm10} again we have
	\begin{eqnarray}\label{thm14:3}
		T_q^\beta(\rho_{P|P_{t+1}\cdots P_{r-1}})
		&\geqslant&\mathcal{K}_\beta T_q^\beta(\rho_{PP_{t+1}})+T_q^\beta(\rho_{P|P_{t+2}\cdots P_{r-1}})\nonumber\\
		&\geqslant& \cdots\nonumber\\
		&\geqslant& \mathcal{K}_\beta\left(\sum_{j=t+1}^{r-2}T_q^\beta(_{PP_j})\right)+T_q^\beta(_{PP_{r-1}}).
	\end{eqnarray}
	
	Combining (\ref{thm14:2}) and (\ref{thm14:3}), we get the inequality (\ref{thm14:1}).
\end{proof}

\begin{remark}
	From Theorem \ref{thm14}, if $k T_q^2(\rho_{PP_j})\geqslant T_q^2(\rho_{P|P_{j+1}\cdots P_{r-1}})$ for all $j=0,1,\cdots, r-2$,  one has
	\begin{eqnarray}\label{thm14:4}
		T_q^\beta(\rho_{P|P_0\cdots P_{r-1}})\geqslant \sum_{j=0}^{r-1}(\mathcal{K}_\beta)^{j}T_q^\beta(\rho_{PP_j}).
	\end{eqnarray}
\end{remark}

Next, using Lemma \ref{lem1} $\mathrm{(b)}$, we further improve the monogamy inequality (\ref{T4}) provided in \cite{Shi2020}.

\begin{lemma}\label{lem2}
	If $T_q^\mu(\rho_{P_1P_3}) \geqslant k T_q^\mu(\rho_{P_1P_2})$, we have for $q\in[\frac{5-\sqrt{13}}{2},\frac{5+\sqrt{13}}{2}]$,
	\begin{eqnarray}\label{lem2:1}
		T_q^\gamma(\rho_{P_1|P_2P_3})\geqslant T_q^\gamma(\rho_{P_1P_2})+ \mathcal{K}_\gamma T_q^\gamma(\rho_{P_1P_3}),
	\end{eqnarray}
	where $\gamma\in[0,\mu]$, $\mu\in[2,\infty)$, $\mathcal{K}_\gamma=\frac{(1+k)^\frac{\gamma}{\mu}-1}{k^{\frac{\gamma}{\mu}}}$ and $k\in[1,\infty)$.
\end{lemma}

\begin{proof}
	From (\ref{T3}), we have $T_q^\mu(\rho_{P_1|P_2P_3})\geqslant T_q^\mu(\rho_{P_1P_2})+T_q^\mu(\rho_{P_1P_3})$ for $\mu\in[2,\infty)$. Hence, we get
	\begin{eqnarray} 	 T_q^\gamma(\rho_{P_1|P_2P_3})&=&(T_q^\mu(\rho_{P_1|P_2P_3}))^{\frac{\gamma}{\mu}}\nonumber\\
		&\geqslant&(T_q^\mu(\rho_{P_1P_2})+T_q^\mu(\rho_{P_1P_3}))^{\frac{\gamma}{\mu}}\nonumber\\
		&=&T_q^\gamma(\rho_{P_1P_2})\Bigg[1+\frac{T_q^\mu(\rho_{P_1P_3})}{T_q^\mu(\rho_{P_1P_2})}\Bigg]^{\frac{\gamma}{\mu}}\nonumber\\
		&\geqslant& T_q^\gamma(\rho_{P_1P_2})\Bigg(1+\mathcal{K}_\gamma\Bigg(\frac{T_q(\rho_{P_1P_3})}{T_q(\rho_{P_1P_2})}\Bigg)^{\gamma}\Bigg)\nonumber\\
		&= &T_q^\gamma(\rho_{P_1P_2})+\mathcal{K}_\gamma T_q^\gamma(\rho_{P_1P_3}).
	\end{eqnarray}
	Here the second inequality is due to Lemma \ref{lem1} $\mathrm{(b)}$.
\end{proof}

Now we generalize our results to multipartite GWV states. The proof is similar to the proof of Theorem \ref{thm14} by using Lemma \ref{lem1} $\mathrm{(b)}$ and Lemma \ref{lem2}.

\begin{theorem}\label{thm6}
	If $k T_q^\mu(\rho_{PP_i})\leqslant T_q^\mu(\rho_{P|P_{i+1}\cdots P_{r-1}})$  for $i=0,1,\cdots, t$, and $T_q^\mu(\rho_{PP_j})\geqslant k T_q^\mu(\rho_{P|P_{j+1}\cdots P_{r-1}})$ for $j=t+1,\cdots, r-2$, $ \forall k\geqslant1$, $0\leqslant t\leqslant r-3$ and $r\geqslant 3$, we have for $q\in [\frac{5-\sqrt{13}}{2},\frac{5+\sqrt{13}}{2}]$,
	\begin{eqnarray}\label{thm6:1}
		T_q^\gamma(\rho_{P|P_0\cdots P_{r-1}})&\geqslant&\sum_{j=0}^{t}(\mathcal{K}_\gamma)^{j}T_q^\gamma(\rho_{PP_j})+(\mathcal{K}_\gamma )^{t+2}\sum_{j=t+1}^{r-2}T_q^\gamma(\rho_{PP_j})\nonumber\\
		&&\ \ \ +(\mathcal{K}_\gamma)^{t+1}T_q^\gamma(\rho_{PP_{r-1}})
	\end{eqnarray}
	with $\gamma\in[0,\mu], \mu \in [2,\infty)$ and $\mathcal{K}_\gamma=\frac{(1+k)^\frac{\gamma}{\mu}-1}{k^\frac{\gamma}{\mu}}$.
\end{theorem}

\begin{remark}
	Since $\frac{(1+k)^\frac{\gamma}{\mu}-1}{k^{\frac{\gamma}{\mu}}}\geqslant 2^{\frac{\gamma}{\mu}}-1$ for $\frac{\gamma}{\mu}\in[0,1]$ and $k\in[1, \infty)$, our new monogamy relation (\ref{thm6:1}) for  T$q$E  is better than (\ref{T4}) given in \cite{Shi2020} which is just a special case of ours for $k=1$. Moreover, the larger the $k$, the tighter the inequality (\ref{thm6:1}).
\end{remark}

\begin{example}\label{exm1}
	Consider the 3-qubit  GW state
	\begin{eqnarray}
		&\ket{\psi}_{A_1A_2A_3}
		=\frac{1}{\sqrt{6}}\ket{100}+\frac{1}{\sqrt{6}}\ket{010}+\frac{2}{\sqrt{6}}\ket{001}.
	\end{eqnarray}
	From the definition of concurrence \cite{Concurrence1}, we get $C(\rho_{A_1|A_2A_3})=\frac{\sqrt{5}}{3}$, $C(\rho_{A_1A_2})=\frac{1}{3}$ and $C(\rho_{A_1A_3})=\frac{2}{3}$. When $q=2$, using (\ref{T8}) we have
	$T_2(\ket{\psi}_{A_1|A_2A_3})=\frac{5}{18}$, $T_2(\rho_{A_1A_2})=\frac{1}{18}$ and $T_2(\rho_{A_1A_3})=\frac{2}{9}$.
	Choosing $\mu=3$, we have $1\leqslant k\leqslant 64$ from Lemma \ref{lem2}. Thus, $T_2^\gamma(\ket{\psi}_{A_1|A_2A_3})\geqslant \left(\frac{1}{18}\right)^\gamma+\frac{(1+k)^\frac{\gamma}{3}-1}
	{k^{\frac{\gamma}{3}}}\left(\frac{2}{9}\right)^\gamma$ from our result (\ref{lem2:1}), and $T_2^\gamma(\ket{\psi}_{A_1|A_2A_3})\geqslant\left(\frac{1}{18}\right)^\gamma
	+\left(2^{\frac{\gamma}{3}}-1\right)\left(\frac{2}{9}\right)^\gamma$ from the result given in \cite{Shi2020}. One can see that our result is better than the result
	in \cite{Shi2020}  for $\gamma\in[0,3]$, and the inequality is tighter as $k$ increases, see Fig. \ref{Fig1}.
\end{example}

\begin{figure}[H]
	\centering
	\includegraphics[width=8cm]{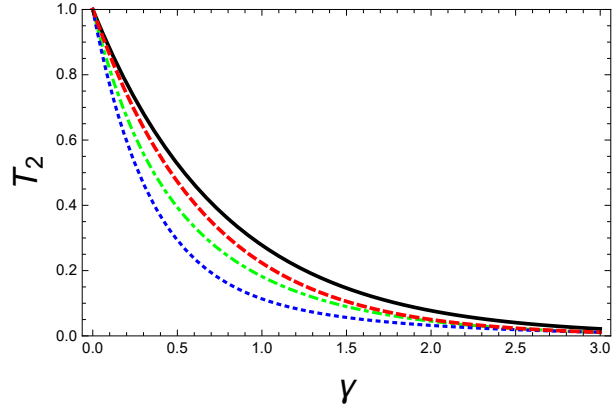}
	\caption{{\small The vertical axis is the the lower bound of the   Tsallis-$q$ entanglement $T_2(\ket{\psi}_{A_1|A_2A_3})$. The  black solid line is the exact values of  $T_2(\ket{\psi}_{A_1|A_2A_3})$. The  red dashed (green dot-dashed) line represents the lower bound from our results for the case of $k=64$ ($k=10$). The blue dotted line represents the lower bound from the result in \cite{Shi2020}.}}
	\label{Fig1}
\end{figure}

\subsection{Tighter  polygamy relations in terms of T$q$EoA}\label{sec3-2}

In this section,  we present the polygamy inequalities for GWV states based on T$q$EoA, which improves the inequality (\ref{T5}).

\begin{theorem}\label{thm7}
	If the subsystems $P_0,P_1,\cdots,P_{r-1}$ satisfy
	\begin{equation}\label{thm7:1}
		k T_q^a(\rho_{PP_j})\geqslant T_q^a(\rho_{PP_{j+1}})\geqslant0
	\end{equation}
	with $j=0,1,\cdots, r-2, 0<k\leqslant 1$, then
	\begin{eqnarray}\label{thm7:2}
		[T_q^a(\rho_{P|P_0\cdots P_{r-1}})]^\mu\leqslant\sum\limits_{j=0}^{r-1}(\mathcal{K}_\mu)^{\omega_H(\overrightarrow{j})}[T_q^a(\rho_{ PP_j})]^\mu,
	\end{eqnarray}
	where $\mu\in (0,1]$, $q\in[\frac{5-\sqrt{13}}{2},2]\cup [3,\frac{5+\sqrt{13}}{2}]$ and  $\mathcal{K}_\mu=\frac{(1+k)^\mu-1}{k^\mu}$.
\end{theorem}

\begin{proof}
	Since $T_q^a(\rho_{P|P_0\cdots P_{r-1}})\leqslant \sum\limits_{j=0}^{r-1} T_q^a(\rho_{PP_j})$ from inequality (\ref{T5}), it is sufficient to prove that
	\begin{equation}\label{thm7:3}
		\left(\sum\limits_{j=0}^{r-1} T_q^a(\rho_{PP_j})\right)^\mu
		\leqslant\sum\limits_{j=0}^{r-1}(\mathcal{K}_\mu)^{\omega_H(\overrightarrow{j})}[T_q^a(\rho_{ PP_j})]^\mu.
	\end{equation}
	
	First, we prove that the inequality (\ref{thm7:3}) holds for the case of $r=2^s$ by using mathematical induction on $s$.
	For $s=1$, using Lemma \ref{lem1} $\mathrm{(c)}$ we have
	\begin{eqnarray}\label{thm7:4}
		(T_q^a(\rho_{PP_0})+T_q^a(\rho_{PP_1}))^\mu
		&=&(T_q^a(\rho_{PP_0}))^\mu \Big(1+\frac{T_q^a(\rho_{PP_1})}{T_q^a(\rho_{PP_0})}\Big)^\mu \nonumber \\
		&\leqslant& (T_q^a(\rho_{PP_0}))^\mu \Bigg[1+\mathcal{K}_\mu\Bigg(\frac{T_q^a(\rho_{PP_1})}{T_q^a(\rho_{PP_0})}\Bigg)^\mu\Bigg] \nonumber \\
		&=&(T_q^a(\rho_{PP_0}))^\mu+\mathcal{K}_\mu(T_q^a(\rho_{PP_1}))^\mu.
	\end{eqnarray}
	Assume that the inequality (\ref{thm7:3}) holds for $r=2^{s-1}$ with $s\geqslant 2$. Consider the case of $r=2^s$.
	From  (\ref{thm7:1}) we have $T_q^a(\rho_{PP_{j+2^{s-1}}})\leqslant k^{2^{s-1}}T_q^a(\rho_{PP_j})$ for $j=0,1,\cdots, 2^{s-1}-1$.
	Therefore,
	\begin{equation*}
		0\leqslant\frac{\sum\nolimits_{j=2^{s-1}}^{2^s-1}T_q^a(\rho_{PP_j})}{\sum\nolimits_{j=0}^{2^{s-1}-1}
			T_q^a(\rho_{PP_j})}\leqslant k^{2^{s-1}}\leqslant k\leqslant1.
	\end{equation*}
	
	Again using Lemma \ref{lem1} $\mathrm{(c)}$, we have
	\begin{eqnarray}\label{thm7:5}
		\Bigg(\sum\limits_{j=0}^{2^s-1}T_q^a(\rho_{PP_j})\Bigg)^\mu
		&=&\Bigg(\sum\limits_{j=0}^{2^{s-1}-1}T_q^a(\rho_{PP_j})\Bigg)^\mu
		\Bigg(1+\frac{\sum_{j=2^{s-1}}^{2^s-1}T_q^a(\rho_{PP_j})}{\sum_{j=0}^{2^{s-1}-1}T_q^a
			(\rho_{PP_j})}\Bigg)^\mu \nonumber\\
		&\leqslant& \Bigg(\sum\limits_{j=0}^{2^{s-1}-1}T_q^a(\rho_{PP_j})\Bigg)^\mu
		\Bigg[1+\mathcal{K}_\mu\Bigg(\frac{\sum_{j=2^{s-1}}^{2^s-1}T_q^a(\rho_{PP_j})}{\sum_{j=0}^{2^{s-1}-1}T_q^a
			(\rho_{PP_j})}\Bigg)^\mu\Bigg] \nonumber\\
		&=&\Bigg(\sum\limits_{j=0}^{2^{s-1}-1}T_q^a(\rho_{PP_j})\Bigg)^\mu+\mathcal{K}_\mu\Bigg(\sum\limits_{j=2^{s-1}}^{2^s-1}T_q^a(\rho_{PP_j})\Bigg)
		^\mu.
	\end{eqnarray}
	From the induction hypothesis, we have
	\begin{equation}\label{thm7:6}
		\Bigg(\sum\limits_{j=0}^{2^{s-1}-1}T_q^a(\rho_{PP_j})\Bigg)^\mu \leqslant
		\sum\limits_{j=0}^{2^{s-1}-1}(\mathcal{K}_\mu)^{\omega_H(\overrightarrow{j})}
		[T_q^a(\rho_{PP_j})]^\mu.
	\end{equation}
	By relabeling the subsystems, we  get
	\begin{equation}\label{thm7:7}
		\Bigg(\sum\limits_{j=2^{s-1}}^{2^s-1}T_q^a(\rho_{PP_j})\Bigg)^\mu \leqslant
		\sum\limits_{j=2^{s-1}}^{2^s-1}(\mathcal{K}_\mu)^{\omega_H(\overrightarrow{j})-1}[T_q^a(\rho_{PP_j})]^\mu.
	\end{equation}
	Hence we have
	\begin{equation}
		\Bigg(\sum\limits_{j=0}^{2^s-1}T_q^a(\rho_{PP_j})\Bigg)^\mu\leqslant
		\sum\limits_{j=0}^{2^s-1}(\mathcal{K}_\mu)^{\omega_H(\overrightarrow{j})}[T_q^a(\rho_{PP_j})]^\mu.
	\end{equation}
	Now we extend the above conclusion to arbitrary integer $r$. Note that there always exists some $s$ such that $0<r\leqslant 2^s$. Let us consider a $(2^s+1)$-partite quantum state
	\begin{equation}\label{thm7:8}
		\gamma_{PP_0P_1\cdots P_{2^s-1}}=\rho_{PP_0P_1\cdots P_{r-1}}\otimes \sigma_{P_r\cdots P_{2^s-1}},
	\end{equation}
	which is the tensor product of $\rho_{PP_0P_1\cdots P_{r-1}}$ and an arbitrary $(2^s-r)$-partite state $\sigma_{P_r\cdots P_{2^s-1}}$.
	
	Similar to the proof above, we have
	\begin{equation}
		[T_q^a(\gamma_{P|P_0P_1\cdots P_{2^s-1}})]^\mu
		\leqslant\sum\limits_{j=0}^{2^s-1}(\mathcal{K}_\mu)^{\omega_H(\overrightarrow{j})}[T_q^a(\gamma_{PP_j})]^\mu,
	\end{equation}
	where $\gamma_{PP_j}$ is the reduced density matrix of $\gamma_{PP_0P_1\cdots P_{2^s-1}}$, $j=0,1,\ldots,2^s-1$.
	
	Taking into account the following obvious facts: $T_q^a\left(\gamma_{P|P_0 P_1 \cdots P_{2^s-1}}\right)=T_q^a\left(\rho_{P|P_0 P_1 \cdots P_{r-1}}\right)$,
	$T_q^a\left(\gamma_{PP_j}\right)=0$ for $j=r, \cdots , 2^s-1$,
	and $T_q^a(\gamma_{PP_j})=T_q^a(\rho_{PP_j})$ for each $j=0, \cdots , r-1$, we get
	\begin{eqnarray}
		[T_q^a(\rho_{P|P_0P_1\cdots P_{r-1}})]^\mu &=&[T_q^a(\gamma_{P|P_0P_1\cdots P_{2^s-1}})]^\mu \nonumber\\
		&\leqslant& \sum\limits_{j=0}^{2^s-1}(\mathcal{K}_\mu)^{\omega_H(\overrightarrow{j})}[T_q^a(\gamma_{PP_j})]^\mu \nonumber\\
		&=& \sum\limits_{j=0}^{r-1}(\mathcal{K}_\mu)^{\omega_H(\overrightarrow{j})}[T_q^a(\rho_{PP_j})]^\mu.
	\end{eqnarray}
	This completes the proof.
\end{proof}

\begin{remark}
	Since $\Big(\frac{(1+k)^\mu-1}{k^\mu}\Big)^{\omega_H(\overrightarrow{j})}\leqslant 1$
	for $\mu\in(0,1]$ and $k\in(0,1]$, our new polygamy relation for T$q$EoA is tighter than inequality (\ref{T5}) in \cite{Shi2020} under certain assumptions for the GWV states . In particular,  when $k=1$, we get a tighter polygamy inequality
	\begin{eqnarray}\label{thm7:10}
		[T_q^a(\rho_{P|P_0\cdots P_{r-1}})]^\mu\leqslant\sum\limits_{j=0}^{r-1}(2^\mu-1)^{\omega_H(\overrightarrow{j})}[T_q^a(\rho_{ PP_j})]^\mu
	\end{eqnarray}
	for any GWV states without the assumptions.
	If one takes $q=2$, our polygamy inequality (\ref{thm7:10}) gives rise to the one in \cite{Shi1}. Furthermore, the inequality (\ref{thm7:2}) gets tighter as $k$ decreases.
\end{remark}

\begin{example}\label{exm2}
	Let us consider the 4-qubit  GW state,
	\begin{eqnarray}
		&\ket{\psi}_{A_1A_2A_3A_4}
		=0.3\ket{0001}+0.4\ket{0010}+{0.5}\ket{0100}+\sqrt{0.5}\ket{1000}.
	\end{eqnarray}
	We have $\rho_{A_1A_2A_3}=0.09\ket{000}\bra{000}+\ket{\phi}\bra{\phi}$ with $\ket{\phi}=0.4\ket{001}+0.5\ket{010}+\sqrt{0.5}\ket{100}$,
	$C(\rho_{A_1A_2})=\frac{\sqrt{2}}{2}$ and $C(\rho_{A_1A_3})=\frac{2\sqrt{2}}{5}$.
	Set $q=2$. We obtain
	$$
	T^a_{2}(\rho_{A_1A_2})=g_2\left(C^2(\rho_{A_1A_2})\right)=\frac{1}{4},~~ T^a_{2}(\rho_{A_1A_3})=g_2\left(C^2(\rho_{A_1A_3})\right)=\frac{4}{25}.
	$$
	Therefore, $[T^a_{2}(\rho_{A_1|A_2A_3})]^\mu\leqslant\left(\frac{1}{4}\right)^\mu+\frac{(1+k)^\mu-1}{k^\mu}\left(\frac{4}{25}\right)^\mu$ from (\ref{thm7:2}), $T^a_{2}(\rho_{A_1|A_2A_3})]^\mu\leqslant\left(\frac{1}{4}\right)^\mu+(2^\mu-1)\left(\frac{4}{25}\right)^\mu$ from (\ref{thm7:10}) and  $[T^a_{2}(\rho_{A_1|A_2A_3})]^\mu\leqslant \left(\frac{1}{4}\right)^\mu+\left(\frac{4}{25}\right)^\mu$ from (\ref{T5}), where $k\in [0.64,1]$ from the condition (\ref{thm7:1}). One can see that our result is better than the ones in \cite{Shi1,Shi2020}, and the smaller the $k$ is, the tighter relation is, see Fig.~\ref{Fig2}.
\end{example}

\begin{figure}[H]
	\centering
	\includegraphics[width=8cm]{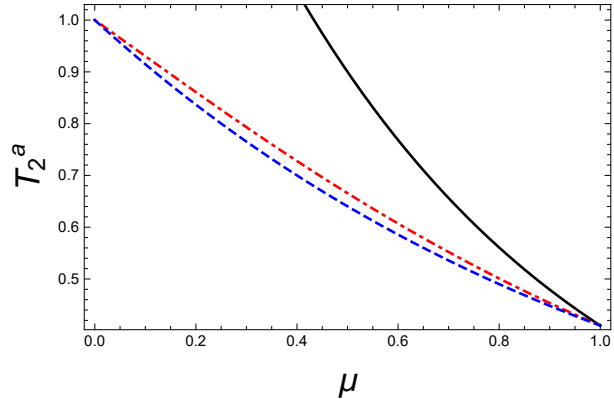}
	\caption{{\small The vertical axis is the upper bound of the Tsallis-$q$ entanglement of assistance $T_2^a(\rho_{A_1|A_2A_3})$. The  blue dashed line represents the upper bound from our result (\ref{thm7:2}) for $k=0.64$, the red dot-dashed line represents the upper bound from the result in \cite{Shi1}, and the black solid line represents the upper bound from the result in \cite{Shi2020}.}}
	\label{Fig2}
\end{figure}

Similar to the improvement from the inequality (\ref{thm10:2}) to the inequality (\ref{thm11:1}), we can analogously improve the polygamy inequality
of Theorem \ref{thm7} under certain conditions.

\begin{theorem}\label{thm8}
	For $q\in[\frac{5-\sqrt{13}}{2},2]\cup [3,\frac{5+\sqrt{13}}{2}]$ we have
	\begin{eqnarray}\label{thm8:1}
		[T_q^a(\rho_{P|P_0\cdots P_{r-1}})]^\mu\leqslant\sum\limits_{j=0}^{r-1}(\mathcal{K}_\mu)^{j}[T_q^a(\rho_{ PP_j})]^\mu
	\end{eqnarray}
	conditioned that
	\begin{equation}\label{thm8:2}
		k T_q^a(\rho_{PP_l})\geqslant \sum\limits_{j=l+1}^{r-1}T_q^a(\rho_{PP_{j}})
	\end{equation}
	for $l=0,1,\cdots, r-2$ and $0<k\leqslant1$, where $\mu\in (0,1]$ and $\mathcal{K}_\mu=\frac{(1+k)^\mu-1}{k^\mu}$.
\end{theorem}
\begin{proof}
	From the inequality (\ref{T5}), we only need to prove
	\begin{equation}\label{thm8:3}
		\Big(\sum_{j=0}^{r-1}T_q^a\left(\rho_{PP_j}\right)\Big)^{\mu}
		\leqslant \sum_{j=0}^{r-1} ( \mathcal{K}_\mu)^{j}[T_q^a(\rho_{PP_j})]^\mu.
	\end{equation}
	We use mathematical induction on $r$ here. It is obvious that the inequality (\ref{thm8:3}) holds for $r=2$ from (\ref{thm7:4}). Assume that it also holds for any positive integer less than $r$. Since $0\leqslant {\sum\limits_{j=1}^{r-1}T_q^a(\rho_{PP_{j}})}/{T_q^a(\rho_{PP_0})}\leqslant k$, we have
	\begin{eqnarray}
		\left(\sum_{j=0}^{r-1}T_q^a\left(\rho_{PP_j}\right)\right)^{\mu}
		&=&[T_q^a(\rho_{PP_0})]^{\mu}
		\Bigg(1+\frac{\sum_{j=1}^{r-1}T_q^a(\rho_{PP_j})}
		{T_q^a(\rho_{PP_0})} \Bigg)^{\mu}\nonumber\\
		&\leqslant& [T_q^a(\rho_{PP_0})]^{\mu}\Bigg[1+ \mathcal{K}_\mu\Bigg(\frac{\sum\nolimits_{j=1}^{r-1}T_q^a(\rho_{PP_{j}})}{T_q^a(\rho_{PP_0})}\Bigg)^{\mu}\Bigg]\nonumber\\
		&=&[T_q^a(\rho_{PP_0})]^{\mu}+\mathcal{K}_\mu\Bigg(\sum\limits_{j=1}^{r-1}T_q^a(\rho_{PP_{j}})\Bigg)^{\mu}\nonumber\\
		&\leqslant& [T_q^a(\rho_{PP_0})]^{\mu}+\mathcal{K}_\mu \sum\limits_{j=1}^{r-1}(\mathcal{K}_\mu)^{j-1}[T_q^a(\rho_{ PP_j})]^\mu \nonumber\\
		&=&\sum\limits_{j=0}^{r-1}(\mathcal{K}_\mu)^{j}[T_q^a(\rho_{ PP_j})]^\mu,
	\end{eqnarray}
	where the first inequality is due to Lemma \ref{lem1} $\mathrm{(c)}$ and the second inequality is due to the induction hypothesis.
\end{proof}

Since
$[T_q^a(\rho_{P|P_0\cdots P_{r-1}})]^\mu\leqslant\sum\limits_{j=0}^{r-1}(\mathcal{K}_\mu)^{j}[T_q^a(\rho_{ PP_j})]^\mu\leqslant\sum\limits_{j=0}^{r-1}(\mathcal{K}_\mu)^{\omega_H(\overrightarrow{j})}[T_q^a(\rho_{ PP_j})]^\mu$
for $\mu\in(0,1]$,  the inequality (\ref{thm8:1}) of Theorem \ref{thm8} is tighter than the inequality (\ref{thm7:2}) of Theorem \ref{thm7} under the conditions.
Similarly, we provide a more general result by changing the conditions of the Theorem \ref{thm8}.

\begin{theorem}\label{thm9:1}
	For $q\in[\frac{5-\sqrt{13}}{2},2]\cup [3,\frac{5+\sqrt{13}}{2}]$ we have
	\begin{eqnarray}\label{thm9:2}
		[T_q^a(\rho_{P|P_0\cdots P_{r-1}})]^\mu&\leqslant& \sum_{j=0}^{t}(\mathcal{K}_\mu)^{j}[T_q^a(\rho_{PP_j})]^\mu+(\mathcal{K}_\mu )^{t+2}\sum_{j=t+1}^{r-2}[T_q^a(\rho_{PP_j})]^\mu\nonumber\\
		&&\ \ \ +(\mathcal{K}_\mu)^{t+1}[T_q^a(\rho_{PP_{r-1}})]^\mu
	\end{eqnarray}
	conditioned that
	$k T_q^a(\rho_{PP_i})\geqslant T_q^a(\rho_{P|P_{i+1}\cdots P_{r-1}})$ for $i=0,1,\cdots, t$ and $T_q^a(\rho_{PP_j})\leqslant k T_q^a(\rho_{P|P_{j+1}\cdots P_{r-1}})$ for $j=t+1,\cdots, r-2$, $ \forall  0<k\leqslant1$, $0\leqslant t\leqslant r-3$, $r\geqslant3$, where $\mu\in (0,1]$ and $\mathcal{K}_\mu=\frac{(1+k)^\mu-1}{k^\mu}$.
\end{theorem}

The proof is similar to the one of Theorem \ref{thm14}, by using the inequality (\ref{thm7:2}) for the case $r=2$ and Lemma \ref{lem1} $\mathrm{(c)}$.

\begin{remark}
	Note that if $k T_q^a(\rho_{PP_j})\geqslant T_q^a(\rho_{P|P_{j+1}\cdots P_{r-1}})$ for all $j=0,1,\cdots, r-2$,  one has
	\begin{eqnarray}\label{thm9:6}
		[T_q^a(\rho_{P|P_0\cdots P_{r-1}})]^\mu\leqslant \sum_{j=0}^{r-1}(\mathcal{K}_\mu)^{j}[T_q^a(\rho_{PP_j})]^\mu.
	\end{eqnarray}
	Due to that $T_q(\rho_{P|P_0\cdots P_{r-1}})=T_q^a(\rho_{P|P_0\cdots P_{r-1}})$ for $q\in[\frac{5-\sqrt{13}}{2}, 2]\cup[3, \frac{5+\sqrt{13}}{2}]$, the above inequalities (\ref{thm7:2}), (\ref{thm8:1}) and (\ref{thm9:2}) also give the upper bounds of $T_q(\rho_{P|P_0\cdots P_{r-1}})$ for GWV states $\ket{\psi}_{A_1\cdots A_n}$.
\end{remark}

\section{Tighter monogamy and polygamy relations based on R$\alpha$E and R$\alpha$EoA for GWV states}\label{sec4}

For a bipartite pure state $|\psi\rangle_{AB}$, the R\'{e}nyi-$\alpha$ entanglement (R$\alpha$E) is defined as \cite{KimR2010} $E_\alpha(|\psi\rangle_{AB})=S_\alpha(\rho_A)$, where $S_\alpha(\rho)= \frac{1}{1-\alpha}\log (\mbox{tr} \rho^\alpha)$ with $\alpha >0$, $\alpha \neq 1$. The $S_{\alpha}(\rho)$ converges to the von Neumann entropy when $\alpha$ tends to 1.
For a bipartite mixed state $\rho_{AB}$, the R\'{e}nyi-$\alpha$ entanglement is given by
\begin{equation*}
	E_{\alpha}\left(\rho_{AB} \right)=\min\limits_{\{p_i, |\psi_i\rangle\}} \sum_i p_i E_{\alpha}(|\psi_i \rangle_{AB}),
\end{equation*}
where the minimum is taken over all possible pure state decompositions of
$\rho_{AB}$.
As a dual concept to R$\alpha$E, the
R\'{e}nyi-$\alpha$ entanglement of assistance (R$\alpha$EoA) is given by
\begin{equation}
	E^{a}_{\alpha}\left(\rho_{AB} \right)=\max \limits_{\{p_i, |\psi_i\rangle\}}\sum_i p_i E_{\alpha}(|\psi_i \rangle_{AB}),
	\label{EoA}
\end{equation}
where the maximum is taken over all possible pure state decompositions of $\rho_{AB}$.

In \cite{KimR2010} the authors have derived an analytical relation between the R\'{e}nyi-$\alpha$ entanglement and concurrence for any two-qubit mixed state $\rho_{AB}$,
\begin{eqnarray}
	E_\alpha  \left( {\rho_{AB} } \right) = f_\alpha  \left[ {C^2 \left( \rho_{AB} \right)} \right],
	\label{q6}
\end{eqnarray}
where $\alpha\in [1,\infty)$ and $f_\alpha \left( x \right)$ has the form
\begin{equation}
	f_\alpha \! \left( x \right)\!= \!\frac{1}{{1 - \alpha }}\!\log _2 \!\left[ {\left( {\frac{{1 \!-\!
					\sqrt {1 - x} }}{2}} \right)^\alpha  \!\!\!\! +\! \left( {\frac{{1 \!+\! \sqrt {1 - x} }}{2}}
		\right)^\alpha  } \right].
	\label{q7}
\end{equation}
Set $x=y^2$ and denote $\tilde{f}_\alpha(y)=f_\alpha(y^2)$. Then the function $\tilde{f}_\alpha(y)$ is monotonically increasing and convex for $y\in[0,1]$. Later, Wang \emph{et al}. \cite{Wang2016} showed that (\ref{q6}) also holds for $\alpha\in[\frac{\sqrt7-1}{2},\infty)$.

Quite recently, Liang {\it et al.} \cite{Liang2020} provided the following analytic formulas for R$\alpha$E and R$\alpha$EoA,
\begin{eqnarray}\label{E1}
	E_\alpha(\rho_{A_{i_1}|A_{i_2}\cdots A_{i_m}})=f_\alpha(C^2(\rho_{A_{i_1}|A_{i_2}\cdots A_{i_m}}))
\end{eqnarray}
for $\alpha\in[\frac{\sqrt 7  - 1}{2}, \infty)$, and
\begin{eqnarray}\label{E2}
	E_\alpha^a(\rho_{A_{i_1}|A_{i_2}\cdots A_{i_m}})=E_\alpha(\rho_{A_{i_1}|A_{i_2}\cdots A_{i_m}})=f_\alpha(C^2(\rho_{A_{i_1}|A_{i_2}\cdots A_{i_m}}))
\end{eqnarray}
for $\alpha\in [\frac{\sqrt7-1}{2}, \frac{\sqrt{13}-1}{2}]$, together with
the following monogamy relation based on R$\alpha$E for GWV states,
\begin{eqnarray}\label{E3}
	E_\alpha^\mu(\rho_{P_1|P_2\cdots P_r})\geqslant \sum_{j=2}^{r} E_\alpha^\mu(\rho_{P_1P_j})
\end{eqnarray}
with $\mu\in[2,\infty)$ and $\alpha\in [\frac{\sqrt7-1}{2}, \infty)$, as well as
the following polygamy inequalities based on R$\alpha$EoA,
\begin{eqnarray}\label{E5}
	(E_\alpha^a(\rho_{P_1|P_2\cdots P_r}))^\mu\leqslant \sum_{j=2}^{r} (E_\alpha^a(\rho_{P_1P_j}))^\mu
\end{eqnarray}
with $\mu\in (0,1]$ and $\alpha\in [\frac{\sqrt7-1}{2}, \frac{\sqrt{13}-1}{2}]$.

Instead of the T$q$E and T$q$EoA used in Theorems of section \ref{sec3}, next we consider the R$\alpha$E and R$\alpha$EoA. The proofs of the Theorems given below are similar to the cases for T$q$E and T$q$EoA.

\subsection{Tighter monogamy  relations in  terms of  R$\alpha$E}\label{sec4-1}
With a similar approach to T$q$E, we first present the following tighter weighted monogamy relations based on R$\alpha$E for GWV states.

\begin{theorem}\label{thm12}
	If the subsystems $P_0,P_1,\cdots,P_{r-1}$ satisfy
	\begin{equation}\label{thm12:1}
		k E_\alpha^2(\rho_{PP_j})\geqslant E_\alpha^2(\rho_{PP_{j+1}})\geqslant 0
	\end{equation}
	for $j=0,1,\cdots, r-2$ and $0<k\leqslant 1$, we have
	\begin{eqnarray}\label{thm12:2}
		E_\alpha^\beta(\rho_{P|P_0\cdots P_{r-1}})\geqslant\sum\limits_{j=0}^{r-1}(\mathcal{K}_\beta)^{\omega_H(\overrightarrow{j})}E_\alpha^\beta(\rho_{ PP_j}),
	\end{eqnarray}
	for $\alpha\in [\frac{\sqrt7-1}{2}, \infty)$, where $\beta\in[2,\infty)$ and $\mathcal{K}_\beta=\frac{(1+k)^\frac{\beta}{2}-1}{k^\frac{\beta}{2}}$.
\end{theorem}

Since $\Big(\mathcal{K}_\beta\Big)^{\omega_H(\overrightarrow{j})}\geqslant 1$ for $\beta\geqslant2$ and $0<k\leqslant 1$, our new polygamy relation for R$\alpha$E is tighter than the inequality (\ref{E3}) in \cite{Liang2020} under certain conditions for the GWV states. Moreover, one can find that the inequality (\ref{thm12:2}) gets tighter as $k$ decreases.

\begin{example}\label{exm3}
	Consider the 3-qubit GW state
	\begin{eqnarray}
		&\ket{\psi}_{A_1A_2A_3}
		=\frac{1}{\sqrt{6}}\ket{100}+\frac{2}{\sqrt{6}}\ket{010}+\frac{1}{\sqrt{6}}\ket{001}.
	\end{eqnarray}
	We have $$C(\ket{\psi}_{A_1|A_2A_3})=\frac{\sqrt{5}}{3}, C(\rho_{A_1A_2})=\frac{2}{3}, C(\rho_{A_1A_3})=\frac{1}{3}.$$
	Choosing $\alpha=2$, from (\ref{E1}) one has
	$$E_2(\ket{\psi}_{A_1|A_2A_3})=\mathrm{log}_2\left(\frac{18}{13}\right),~~ E_2(\rho_{A_1A_2})=\mathrm{log}_2\left(\frac{9}{7}\right), ~~ E_2(\rho_{A_1A_3})=\mathrm{log}_2\left(\frac{18}{17}\right).$$
	Then  $E_2^\beta(\ket{\psi}_{A_1|A_2A_3})\geqslant\Big[\mathrm{log}_2\Big(\frac{9}{7}\Big)\Big]^\beta+\frac{(1+k)^\frac{\beta}{2}-1}
	{k^{\frac{\beta}{2}}}\Big[\mathrm{log}_2\Big(\frac{18}{17}\Big)\Big]^\beta$ from our result (\ref{thm12:2}),  and $E_2^\beta(\ket{\psi}_{A_1|A_2A_3})\geqslant\Big[\mathrm{log}_2\Big(\frac{9}{7}\Big)\Big]^\beta+\Big[\mathrm{log}_2\Big(\frac{18}{17}\Big)\Big]^\beta$ from the result (\ref{E3}), where $0.52\leqslant k\leqslant1$. One can see that our result is better than the result (\ref{E3})  in \cite{Liang2020} for $\beta\geqslant2$, see Fig.~\ref{Fig3}.
\end{example}

\begin{figure}[H]
	\centering
	\includegraphics[width=8cm]{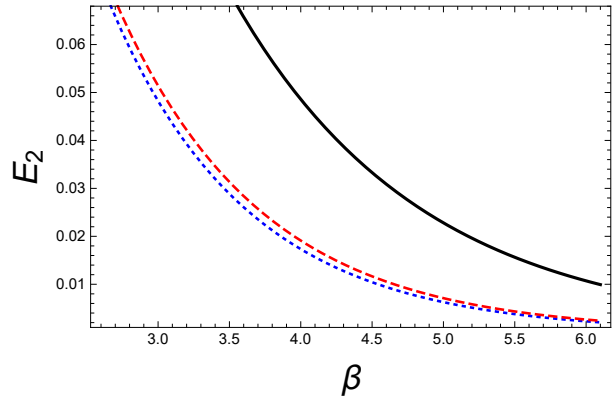}
	\caption{\small The vertical axis is the lower bound of the R\'{e}nyi-$\alpha$ entanglement  $E_2(\ket{\psi}_{A_1|A_2A_3})$. The solid black line represents the exact values of  $E_2(\ket{\psi}_{A_1|A_2A_3})$, the  red dot-dashed line represents the lower bound from our results for $k=0.52$, and the dashed blue line represents the lower bound from the result (\ref{E3}) in \cite{Liang2020}.}
	\label{Fig3}
\end{figure}

The inequality (\ref{thm12:2}) can be further improved under certain conditions.

\begin{theorem}\label{thm13}
	When $\alpha\in [\frac{\sqrt7-1}{2}, \infty)$, we have
	\begin{eqnarray}\label{thm13:1}
		E_\alpha^\beta(\rho_{P|P_0\cdots P_{r-1}})\geqslant\sum\limits_{j=0}^{r-1}(\mathcal{K}_\beta)^{j}E_\alpha^\beta(\rho_{ PP_j})
	\end{eqnarray}
	conditioned that
	\begin{equation}\label{thm13:2}
		k E_\alpha^2(\rho_{PP_l})\geqslant \sum\limits_{j=l+1}^{r-1}E_\alpha^2(\rho_{PP_{j}})
	\end{equation}
	for $l=0,1,\cdots, r-2$ and $0<k\leqslant 1$, where  $\beta\in[2,\infty)$ and $\mathcal{K}_\beta=\frac{(1+k)^\frac{\beta}{2}-1}{k^\frac{\beta}{2}}$.
\end{theorem}

In general, we have the following monogamy inequality.

\begin{theorem}\label{thm15}
	When $\alpha\in [\frac{\sqrt7-1}{2}, \infty)$, we have
	\begin{eqnarray}\label{thm15:1}
		E_\alpha^\beta(\rho_{P|P_0\cdots P_{r-1}})&\geqslant& \sum_{j=0}^{t}(\mathcal{K}_\beta)^{j}E_\alpha^\beta(\rho_{PP_j})+(\mathcal{K}_\beta )^{t+2}\sum_{j=t+1}^{r-2}E_\alpha^\beta(\rho_{PP_j})\nonumber\\
		&&\ \ \ +(\mathcal{K}_\beta)^{t+1}E_\alpha^\beta(\rho_{PP_{r-1}})
	\end{eqnarray}
	conditioned that
	$k E_\alpha^2(\rho_{PP_i})\geqslant E_\alpha^2(\rho_{P|P_{i+1}\cdots P_{r-1}})$ for $i=0,1,\cdots, t$ and $E_\alpha^2(\rho_{PP_j})\leqslant k E_\alpha^2(\rho_{P|P_{j+1}\cdots P_{r-1}})$ for $j=t+1,\cdots, r-2$, $ \forall  0<k\leqslant1$, $0\leqslant t\leqslant r-3$ and $r\geqslant3$, where $\beta\in[2,\infty)$ and $\mathcal{K}_\beta=\frac{(1+k)^\frac{\beta}{2}-1}{k^\frac{\beta}{2}}$.
\end{theorem}

\begin{remark}
	If $k E_\alpha^2(\rho_{PP_j})\geqslant E_\alpha^2(\rho_{P|P_{j+1}\cdots P_{r-1}})$ for all $j=0,1,\cdots, r-2$,  we have
	\begin{eqnarray}\label{thm15:4}
		E_\alpha^\beta(\rho_{P|P_0\cdots P_{r-1}})\geqslant \sum_{j=0}^{r-1}(\mathcal{K}_\beta)^{j}E_\alpha^\beta(\rho_{PP_j}).
	\end{eqnarray}
\end{remark}

\subsection{Tighter  polygamy relations  in terms of R$\alpha$EoA }\label{sec4-2}

Now we establish the tighter polygamy relations for R$\alpha$EoA by using a similar approach to T$q$EoA.

\begin{theorem}\label{thm3}
	If the subsystems $P_0,P_1,\cdots,P_{r-1}$ satisfy
	\begin{equation}\label{thm3:1}
		k E_\alpha^a(\rho_{PP_j})\geqslant E_\alpha^a(\rho_{PP_{j+1}})\geqslant0
	\end{equation}
	for $j=0,1,\cdots, r-2$ and $0<k\leqslant1$, we have
	\begin{eqnarray}\label{thm3:2}
		[E_\alpha^a(\rho_{P|P_0\cdots P_{r-1}})]^\mu\leqslant\sum\limits_{j=0}^{r-1}(\mathcal{K}_\mu)^{\omega_H(\overrightarrow{j})}[E_{\alpha}^a(\rho_{ PP_j})]^\mu,
	\end{eqnarray}
	where $\mu\in(0,1]$, $\alpha\in[\frac{\sqrt7-1}{2}, \frac{\sqrt{13}-1}{2}]$ and  $\mathcal{K}_\mu=\frac{(1+k)^\mu-1}{k^\mu}$.
\end{theorem}

Since  $\Big(\mathcal{K}_\mu\Big)^{\omega_H(\overrightarrow{j})}\leqslant 1$
for $\mu\in(0,1]$ and $k\in(0,1]$, our new polygamy inequality for R$\alpha$EoA is tighter than the inequality (\ref{E5}) in \cite{Liang2020} under certain conditions for the GWV states.
Also, one finds that the smaller the $k$ is, the tighter the inequality (\ref{thm3:2}) is.

\begin{example}\label{exm4}
	Let us again consider the 4-qubit  GW state presented in Example \ref{exm2}.
	Choosing $\alpha=1.2$ we have
	$$E^a_{1.2}(\rho_{A_1A_2})=f_{1.2}\Big[\Big(\frac{\sqrt{2}}{2}\Big)^2\Big]\approx 0.549339, ~E^a_{1.2}(\rho_{A_1A_3})=f_{1.2}\Big[\Big(\frac{2\sqrt{2}}{5}\Big)^2\Big]\approx 0.372954.$$
	From (\ref{thm3:1}), we get $k\in [0.68, 1]$. Then our inequality (\ref{thm3:2}) yields  that $[E_{1.2}^a(\rho_{A_1|A_2A_3})]^\mu \leqslant 0.549339^\mu+\frac{(1+k)^\mu-1}{k^\mu}0.372954^\mu$, and $[E_{1.2}^a(\rho_{A_1|A_2A_3})]^\mu \leqslant 0.549339^\mu+(2^\mu-1)0.372954^\mu$ for $k=1$. While, the inequality (\ref{E5}) yields $[E_{1.2}^a(\rho_{A_1|A_2A_3})]^\mu \leqslant 0.549339^\mu+0.372954^\mu$. Hence, our results are better than one in \cite{Liang2020}, and the inequality gets tighter as $k$ decreases, see Fig. \ref{Fig4}.
\end{example}

\begin{figure}[H]
	\centering
	\includegraphics[width=8cm]{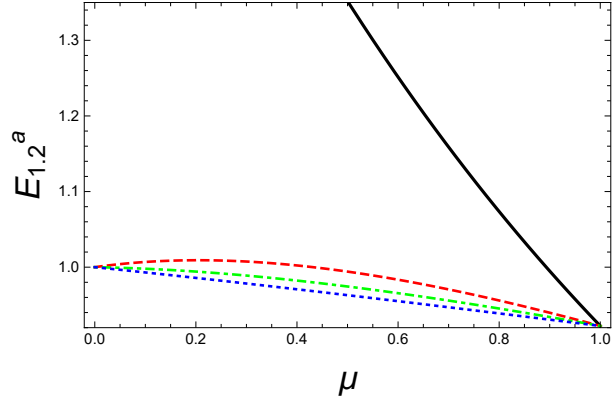}
	\caption{The vertical axis is the upper bound of the R\'{e}nyi-$\alpha$ entanglement of assistance $E_{1.2}^a(\rho_{A_1|A_2A_3})$. The red dashed (green dot-dashed, blue dotted) line represents the upper bound from our result (\ref{thm3:2}) for $k=1$ ($k=0.8$, $k=0.7$), and the  black solid  line represents the upper bound from (\ref{E5}) in \cite{Liang2020}.}
	\label{Fig4}
\end{figure}

Analogously we can improve the polygamy inequality of Theorem \ref{thm3} under certain conditions.

\begin{theorem}\label{thm4}
	For $\alpha\in[\frac{\sqrt7-1}{2},\frac{\sqrt{13}-1}{2}]$ we have
	\begin{eqnarray}\label{thm4:0}
		[E_\alpha^a(\rho_{P|P_0\cdots P_{r-1}})]^\mu\leqslant\sum\limits_{j=0}^{r-1}(\mathcal{K}_\mu)^{j}[E_{\alpha}^a(\rho_{ PP_j})]^\mu
	\end{eqnarray}
	conditioned that
	\begin{equation}\label{thm4:1}
		k E_\alpha^a(\rho_{PP_l})\geqslant \sum\limits_{j=l+1}^{r-1}E_\alpha^a(\rho_{PP_{j}}),
	\end{equation}
	for $l=0,1,\cdots, r-2, 0<k\leqslant1$, where $\mu\in (0,1]$ and $\mathcal{K}_\mu=\frac{(1+k)^\mu-1}{k^\mu}$.
\end{theorem}

Due to $\omega_H(\overrightarrow{j})\leqslant j$, one has $\sum\limits_{j=0}^{r-1}\Big(\mathcal{K}_\mu\Big)^{j}[E_\alpha^a(\rho_{ PP_j})]^\mu\leqslant\sum\limits_{j=0}^{r-1}\Big(\mathcal{K}_\mu\Big)^{\omega_H(\overrightarrow{j})}[E_\alpha^a(\rho_{ PP_j})]^\mu$
for $\mu\in(0,1]$ and $k\in(0,1]$. Therefore, the inequality (\ref{thm4:0}) of Theorem \ref{thm4} is tighter than the inequality (\ref{thm3:2}) of Theorem \ref{thm3}.

Similar to the case of T$q$EoA, we also have the following polygamy relation for R$\alpha$EoA under certain conditions.

\begin{theorem}\label{thm4:3}
	For $\alpha\in[\frac{\sqrt7-1}{2},\frac{\sqrt{13}-1}{2}]$ we have
	\begin{eqnarray}\label{thm4:4}
		[E_\alpha^a(\rho_{P|P_0\cdots P_{r-1}})]^\mu&\leqslant& \sum_{j=0}^{t}(\mathcal{K}_\mu)^{j}[E_\alpha^a(\rho_{PP_j})]^\mu+(\mathcal{K}_\mu )^{t+2}\sum_{j=t+1}^{r-2}[E_\alpha^a(\rho_{PP_j})]^\mu\nonumber\\
		&&\ \ \ +(\mathcal{K}_\mu)^{t+1}[E_\alpha^a(\rho_{PP_{r-1}})]^\mu
	\end{eqnarray}
	conditioned that
	$k E_\alpha^a(\rho_{PP_i})\geqslant E_\alpha^a(\rho_{P|P_{i+1}\cdots P_{r-1}})$ for $i=0,1,\cdots, t$ and $E_\alpha^a(\rho_{PP_j})\leqslant k E_\alpha^a(\rho_{P|P_{j+1}\cdots P_{r-1}})$ for $j=t+1,\cdots, r-2$, $ \forall  0<k\leqslant1, 0\leqslant t\leqslant r-3, r\geqslant 3$, where $\mu\in (0,1]$ and $\mathcal{K}_\mu=\frac{(1+k)^\mu-1}{k^\mu}$.
\end{theorem}

\begin{remark}
	If $k E_\alpha^a(\rho_{PP_j})\geqslant E_\alpha^a(\rho_{P|P_{j+1}\cdots P_{r-1}})$ for all $j=0,1,\cdots, r-2$, then
	\begin{eqnarray}\label{thm9:6}
		[E_\alpha^a(\rho_{P|P_0\cdots P_{r-1}})]^\mu\leqslant \sum_{j=0}^{r-1}(\mathcal{K}_\mu)^{j}[E_\alpha^a(\rho_{PP_j})]^\mu.
	\end{eqnarray}
	Since $E_\alpha(\rho_{P|P_0\cdots P_{r-1}})=E_\alpha^a(\rho_{P|P_0\cdots P_{r-1}})$ for $\alpha\in[\frac{\sqrt7-1}{2},\frac{\sqrt{13}-1}{2}]$, the above inequalities (\ref{thm3:2}), (\ref{thm4:0}) and (\ref{thm4:4}) are also upper bounds of $E_\alpha(\rho_{P|P_0\cdots P_{r-1}})$ for GWV states $\ket{\psi}_{A_1\cdots A_n}$.
\end{remark}

\section{Conclusion}\label{sec5}
Both monogamy and polygamy relations of quantum entanglement are the fundamental properties of multipartite entangled states. We have investigated the monogamy
and polygamy relations of multipartite entanglement for the arbitrary $n$-qudit GWV states with respect to different partitions. By using the Hamming weight of the binary vectors related to the partition of the subsystems, we have established a class of monogamy inequalities in terms of the $\beta$th power of T$q$E for the GWV states when $\beta\geqslant2$, as well as the polygamy inequalities in terms of the $\mu$th power of T$q$EoA when $0<\mu\leqslant1$. Similarly, we have also provided the monogamy and polygamy relations based on E$\alpha$E and E$\alpha$EoA for the GWV states. We
have further shown that our monogamy and polygamy
inequalities hold  under some conditions for the GWV states  in a tighter way than the existing ones and can also recover the previous relations,  thus they give rise to better restrictions on entanglement distribution among the subsystems of the GWV states.

\begin{acknowledgements}
This work is supported by NSFC (Grant No. 12075159), Beijing Natural Science Foundation (Z190005), Academy for Multidisciplinary Studies, Capital Normal University, the Academician Innovation Platform of Hainan Province, and Shenzhen Institute for Quantum Science and Engineering, Southern University of Science and Technology (No. SIQSE202001).
\end{acknowledgements}


\end{document}